\documentclass[a4paper,twocolumn,11pt,accepted=2017-05-09]{quantumarticle}
\pdfoutput=1
\usepackage[utf8]{inputenc}
\usepackage[english]{babel}
\usepackage[T1]{fontenc}
\usepackage{amsmath}
\usepackage{hyperref}
\usepackage{braket}
\usepackage{dsfont}
\usepackage{etoolbox}

\usepackage{tikz}
\usetikzlibrary{quantikz}
\usepackage{lipsum}

\usepackage{amsthm, amssymb}
\newtheorem{Theorem}{Theorem}

\newcommand{\bigO}{\mathcal{O}}

\begin{document}

\title{Quantum algorithm for Ewald summation based computation of long-range electrostatics}

\author{Mansur Ziiatdinov}
\affiliation{MIFT Department, Università degli Studi di Messina, Messina 98122, Italy}

\author{Igor Novikov}
\affiliation{Chemistry Department, University of Massachusetts, Lowell, MA 01854, United States}

\author{Farid Ablayev}
\affiliation{Institute of Computational Mathematics and Information Technologies, Kazan Federal University, 420008, Russia}

\author{Valeri Barsegov}
\affiliation{Chemistry Department, University of Massachusetts, Lowell, MA 01854, United States}
\email{Valeri\_Barsegov@uml.edu}

\maketitle

\begin{abstract}
  In computational molecular science, calculation of electrostatic interactions involving charged atoms -- the strongest interactions in condensed phases, is a major bottleneck. We propose a quantum-classical algorithm for fast, yet, accurate computation of the Coulomb electrostatic energy for a system of point charges. The algorithm employs the Ewald method based decomposition of electrostatic energy into several energy terms, of which the ``Fourier component'' (long-range electrostatics) computed on a quantum device, utilizing the power of Quantum Fourier Transform (QFT). We demonstrate that the algorithm complexity is {$\mathbf{N \bf{log} M}$} and that the quantum advantage for a system of point charges in the three-dimensional space is achieved when the number of grid points $\mathbf{M^3}$ exceeds the number of charges $\mathbf{N}$. The numerical error is small $\mathbf{<10^{-3}}$. The algorithm can be implemented to run the all-atom Molecular Dynamics simulations on a quantum device requiring 15 qubits, thereby expanding the scope of applications of QFT-based methods to computational chemistry and biophysics.
\end{abstract}


\section{Introduction}
Modern biological science faces a pressing need to solve fundamental problems, including protein folding, dynamics of DNA and RNA, drug discovery, genome assembly, cell division, extracellular matrix remodeling to name a few \cite{Protein_folding, Cell_division, DNA_dynamics}; yet, these problems cannot be solved using experiments alone. Consequently, computational exploration connecting biology, chemistry and physics has become an important tool enabling researchers to use the power of computers to describe biological processes \cite{Separat_of_domains, Sequence_comparison, Lattice_microbes}. However, numerical exploration of real large-size biological systems (e.g. genome, viruses, cells, etc.) requires computational power far exceeding that of modern classical computers. Quantum computers offer a unique opportunity to solve numerical problems that are not feasible on classical computers~\cite{arute2019quantum,wu2021strong,hempel2018quantum}.

Molecular science is a key application area for quantum computing. Molecular Dynamics (MD) simulations are an important numerical tool to explore the properties of biological macromolecules (RNA, DNA, and proteins)~\cite{karplus2002molecular, MD_perspective}. Due to charged nature of these molecules and to polar aqueous environments that host these molecules \cite{Electrostatics, Electrostatic_calculations}, it is necessary to describe electrostatic interactions between charged atoms -- the strongest interactions in condensed phases. This task is accomplished by splitting electrostatic forces (or energies) into the short-range contribution and the long-range contribution, with the latter being the main computational bottleneck \cite{Com_sim_of_liquids}. The Ewald summation method~\cite{ewald1921evaluation} performs an excellent job at splitting the slowly converging sum over the Coulomb energies (or forces) into the sums that converge exponentially fast. Yet, the Ewald method remains computationally demanding because a part of the problem (long-range electrostatics) is solved in the reciprocal space under Periodic Boundary Conditions (PBC), which requires using the Fourier transformation~\cite{Pollock_1996_Comments_P3M_FMM_Ewald}.

The point-charge description of atoms is widely used in current state-of-art force fields for biological molecules, including OPLS~\cite{jorgensen1996opls}, Charmm~\cite{brooks2009charmm}, GROMOS~\cite{GROMOS} and Amber~\cite{salomon2013amber}. The point charges are described by the Coulomb interaction potential, which shows a slow $\sim 1/r$ decay with the interparticle separation distance $r$. The particle-particle-particle-mesh (P$^3$M) method~\cite{eastwood1980p3m3dp}, the particle mesh Ewald (PME) method~\cite{darden1993pme}, and the smooth particle mesh Ewald (SPME) method~\cite{essmann1995spme} all implement fast Fourier transform (FFT)~\cite{Pollock_1996_Comments_P3M_FMM_Ewald}. This enables one to reduce the computational complexity for the reciprocal part of the Ewald sum estimation (long-range electrostatics) to the order $N$$\text{log}N$. If the real space distance cutoff is small enough, then the $N$$\text{log}N$ scaling law also applies to the complete Ewald sum~\cite{deserno1998_mesh_ewald_i}. Although FFT is a grid transformation, discretization conflicts can be resolved and associated discretization errors can be mitigated~\cite{deserno1998_mesh_ewald_ii}.

Here, we propose a new quantum algorithm for the Ewald summation-based calculation of long-range electrostatics, which utilizes the Quantum Fourier Transform (QFT)\cite{nielsen2010quantum}. Current NISQ (Noisy Intermediate-Scale Quantum) era quantum computers are characterized by a limited number of qubits and significant noise, and so these computers are not yet fault-tolerant or scalable enough to achieve a full quantum advantage. For this reason, we propose an algorithm, which takes advantage of both the Quantum Processing Units (QPUs) for computation of the long-range electrostatics, and classical Central Processing Units (CPUs) for calculation of the short-range electrostatics as well as self-interaction and dipole interaction contributions. Because quantum devices that will become available in the near future will likely have limitations in the types of arithmetic operations, here we restrict our hybrid quantum-classical algorithm to one- and two-qubit operations. 
\section{Methods}
\subsection{Quantum computation}\label{quantum-basics}

\noindent
{\bf Quantum state:} While a classical computer with \(n\) bits can exist in one of the \(2^n\) possible states, a quantum computer with \(n\) quantum bits (qubits) can exist in all of these states simultaneously. A state of quantum computer $\ket{\psi}$ can be described as a complex-valued \(2^n\)-dimensional vector \(\ket{\psi} = \sum_{j=0}^{N-1} \alpha_j \ket{j}\), where $\{\ket{j}\}$ are the normalized eigenstates of the system's Hamiltonian $H$ (tensor products of the basis states of \(n\) qubits), and $\alpha_{j}$ are constant coefficients.

\noindent
{\bf The Hadamard gate $H$} is a single-qubit operator, which can be represented by the 2$\times$2 matrix:
$$H = \frac{1}{\sqrt{2}} \begin{bmatrix} 1 & 1 \\ 1 & -1 \end{bmatrix}$$
$H$ transforms the eigenstate $\ket{0}$ into a mixed state $\frac{1}{\sqrt{2}}(\ket{0} + \ket{1})$, and the eigenstate $\ket{1}$ into a mixed state $\frac{1}{\sqrt{2}}(\ket{0} - \ket{1})$. If we apply $H$ to state $\ket{0}$ or state $\ket{1}$ and perform a measurement, then we have an equal probability (1/2) of observing $\ket{0}$ or $\ket{1}$.

\noindent {\bf The $R_n$ gate} (phase gate) is defined by the 2$\times$2 matrix,
$$R_n = \begin{bmatrix} 1 & 0 \\ 0 & e^{2\pi i/2^n} \end{bmatrix}$$
The action of $R_n$ gate is multiplication of the amplitude of state $\ket{1}$ by a phase factor $e^{2\pi i/2^n}$. In particular,
$$R_1 = Z = \begin{bmatrix} 1 & 0 \\ 0 & -1 \end{bmatrix}$$
and
$$R_2 = S = \begin{bmatrix} 1 & 0 \\ 0 & i \end{bmatrix}$$
For large $n$, $e^{2\pi i/2^n}$ is close to the unity, and $R_n$ gate is close to the identity gate {\bf I} (2$\times$2 identity matrix). 

\noindent {\bf The controlled-$R_n$} applies $R_n$ to a target qubit only if a control qubit is in the $\ket{1}$ state. In the matrix representation: 

$$\text{controlled-}R_n = \begin{bmatrix}
    1 & 0 & 0 & 0 \\
    0 & 1 & 0 & 0 \\
    0 & 0 & 1 & 0 \\
    0 & 0 & 0 & e^{2\pi i/2^n}
    \end{bmatrix}
    $$

\noindent 
{\bf Measurements:} During computations, the quantum computer remains in the state of superposition of the basis states. To obtain a result of computations, it is necessary to perform a measurement, which returns one of the basis states with a probability equal to the square of the amplitude of that state. That is, if a quantum computer is in a mixed state \(\ket{\psi} = \sum_{j=0}^{N-1} \alpha_j \ket{j}\), where \(\sum_{j=0}^{N-1} \alpha_j^2 = 1\), then, after the measurement \(\mathcal{M}\), the state collapses into one of the basis states \(\ket{j}\) with the probability \(\alpha_j^2\).

\subsection{Quantum state initialization}\label{state-preparation}

Quantum computations usually start with a basis state \(\ket{0^n} = \ket{0} \otimes \ket{0} \otimes \dots \otimes \ket{0}\). It is necessary to prepare (or initialize) an arbitrary state from the state \(\ket{0^n}\). Given state \(\ket{\psi}\), one needs to construct a unitary gate \(U\) such that \(U\ket{0^n} = \ket{\psi}\). A general approach for the state initialization is to employ the M\"ott\"onen method (used by default in Qiskit \cite{Qiskit_IBM}). It requires \(2^{n+2} - 4n - 4\) CNOT gates and \(2^{n+2}-5\) one-qubit elementary rotations to prepare a (generic) state of \(n\) qubits \cite{Mottonen2004_states}. More efficient state initialization methods are available if a state is not generic, e.g. if it contains only a few non-zero amplitudes, in which case the Gleinig-Hoefler method~\cite{Gleinig2021_sparse_states} can be used. Using this method, the initialization of a sparse state of \(n\) qubits, which has only \(S\) non-zero amplitudes, requires \(\bigO(Sn)\) gates.

\subsection{Quantum Fourier Transform}\label{qft}
\noindent
{\bf The QFT algorithm} computes the Fourier transform of a mixed state \(\ket{\psi} = \sum_{j=0}^{N-1} \alpha_{j}\ket{j} \), and outputs the mixed state in the Fourier domain, \(\ket{\phi} = \sum_{k=0}^{N-1} \beta_{k}\ket{k}\), with the coefficients
\[
  \beta_{k} = \frac{1}{\sqrt{N}} \sum_{j=0}^{N-1} \alpha_{j} \exp\Big(\frac{2\pi i j k}{N}\Big)
\]
The multidimensional version of Fourier transform \cite{garcia-molina_rodriguez-mediavilla2022} is
\begin{equation}\label{eq:multi-qft}
\ket{\mathbf{r}} \mapsto \frac{1}{\sqrt{N}} \sum_{\mathbf{s}} e^{i 2 \pi \mathbf{r} \cdot \mathbf{s}} \ket{\mathbf{s}},
\end{equation}
where \(\ket{\mathbf{r}}\) is a \(d\)-dimensional vector encoded as a collection of \(d\) shifts. The algorithm uses \(\bigO(d(\log N)^{2})\) gates to compute the quantum Fourier transformation.

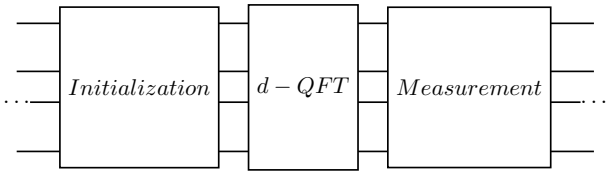
\begin{figure}[h]
  \centering
  \resizebox{\columnwidth}{!}{
      \begin{quantikz}
        \lstick{} & \gate[4]{Initialization} & \gate[4]{d-QFT} & \gate[4]{Measurement} & \qw\\
        \lstick{} & & & & \qw\\
        \dots &&&& \qw\dots\\
        \lstick{} & & & & \qw\\
      \end{quantikz}
  }
  \caption{The workflow of computation of the $E^{L}$-component of electrostatic energy $E$ (Eq. (\ref{eq:eq5})), which consists of the following three steps: the state initiation, the $d$-dimensional QFT, and the measurement.}
  \label{fig:overall-pipeline}
\end{figure}

\subsection{Numerical calculations}\label{sec:method-description}
\noindent
The total electrostatic interaction energy for a system of point charges, $E = E^{S} + E^{L} + E^{self}+ E^{dip}$, involves contributions from the energy terms $E^{S}$, $E^{L}$, $E^{self}$, and $E^{dip}$ (see Eq. (\ref{eq:eq5}) below). In the algorithm proposed, the $E^{L}$-contribution is calculated in a quantum computer, while the other energy terms are calculated in a classical computer (see Results). The calculation of \(E^L\) consists of 1) the state initialization, ii) the multidimensional Fourier transform, and iii) the measurement (displayed schematically in Fig. 1), with the scaling law that depends on different parameters. For example, the state initialization step depends on the number of charges \(N\), whereas the multidimensional Fourier transform depends on the number of grid points \(M\). In this study, we used standard state initialization and state measurement techniques implemented in Qiskit \cite{Qiskit_IBM}.

\section{Results}

\subsection{Ewald summation approach}\label{sec:problem-description}

Consider a system of $N$ charged particles $q_1,q_2,\ldots,q_N$ in vacuum with coordinates ${\bf{r}}_1, {\bf{r}}_2,\ldots,{\bf{r}}_N$. The total energy of the system of particles due to the electrostatic (Coulomb) interactions is given by
\begin{equation}
E=\frac{1}{4\pi\epsilon_0}\sum_{(i,j)}\frac{q_i q_j}{\|\bf{r}_{ij}\|}
\label{eq:eq3}
\end{equation}
where $\bf{r}_{ij}=\bf{r}_j-\bf{r}_i$ is the vector connecting the $i$-th and $j$-th particles, and $\epsilon_0$ is the dielectric constant. In Eq. (\ref{eq:eq3}), the summation is performed over all pairs $(i,j)$ of particles. We impose the periodic boundary conditions (PBC), described by the repeat vectors ${\bf{c}}_1,{\bf{c}}_2,{\bf{c}}_3$ forming a subcell. This implies that when there is a charge $q_i$ at location ${\bf{r}}_i$, there are also charges $q_i$ at locations in the image cells ${\bf{r}_i}+n_1{\bf{c}_1}+n_2{\bf{c}_2}+n_3{\bf{c}_3}$, where $n_1$, $n_2$, and $n_3$ are integers. The PBC is invoked to evaluate the double summation over $i\neq j$ in Eq. (\ref{eq:eq3}) above. The Coulomb interaction energy of the supercell (i.e. subcell plus the image cells) forming a simple cubic lattice $(n_1,n_2,n_3)$ under the PBC condition can be recast in the following form:
\begin{equation}
\begin{aligned}
    & E=\frac{1}{4\pi\epsilon_0}\sum_{(i,j)}\frac{q_i q_j}{\|\bf{r}_{ij}\|} \\
    & =\frac{1}{4\pi\epsilon_0}\sum_{\bf n}\sum_{(i,j)}\frac{q_i q_j}{\|\bf{r}_{ij}+{\bf n}L\|}
    \label{eq:eq4}
\end{aligned}
\end{equation}
where ${\bf n}L=n_1{\bf{c}_1}+n_2{\bf{c}_2}+n_3{\bf{c}_3}$, and $L=\|{\bf c}_1\|=\|{\bf c}_2\|=\|{\bf c}_3\|$ is the cell length. The summation in Eq. (\ref{eq:eq4}) shows very slow convergence and is only conditionally convergent.

The problem of slow conditional convergence of the double summation in Eq. (\ref{eq:eq4}) can be overcome by using the Ewald summation method~\cite{ewald1921evaluation}. The expression for $E$ can be recast into the following contributions, each converging rapidly and absolutely:

\begin{equation}
\begin{aligned}
  E &= E^{S} + E^{L} + E^{self} + E^{dip} \\
  &= \frac{1}{8\pi \varepsilon_{0}}\sum_{\textbf{n}}\sum_{i=1}^{N}\sum_{j=1}^{N}\phantom{}^{'}\frac{q_{i}q_{j}}{|\textbf{r}_{i} - \textbf{r}_{j} + \textbf{n}L|} \\
  &\quad \times\mathrm{Erfc} \Big( \frac{|\textbf{r}_{i} - \textbf{r}_{j} + \textbf{n}L|}{\sqrt{2}\sigma} \Big) \\
  &+ \frac{1}{2V\varepsilon_{0}}\sum_{\textbf{k$\neq$0}}\frac{e^{-\sigma^{2}k^{2}/2}}{k^{2}}|S(\mathrm{\textbf{k}})|^{2} \\ &\quad - \frac{1}{4\sqrt{2}\pi^{3/2}\sigma \varepsilon_{0}}\sum_{i=1}^{N}q_{i}^{2}\\
  &\quad + \frac{1}{2\varepsilon_{0}(1+2\epsilon')V}\Big(\sum_{i}^Nq_{i}\textbf{r}_{i}\Big)^2
\end{aligned}
\label{eq:eq5}
\end{equation}

\begin{figure}[ht]
    \centering
    \includegraphics[width=0.41\textwidth]{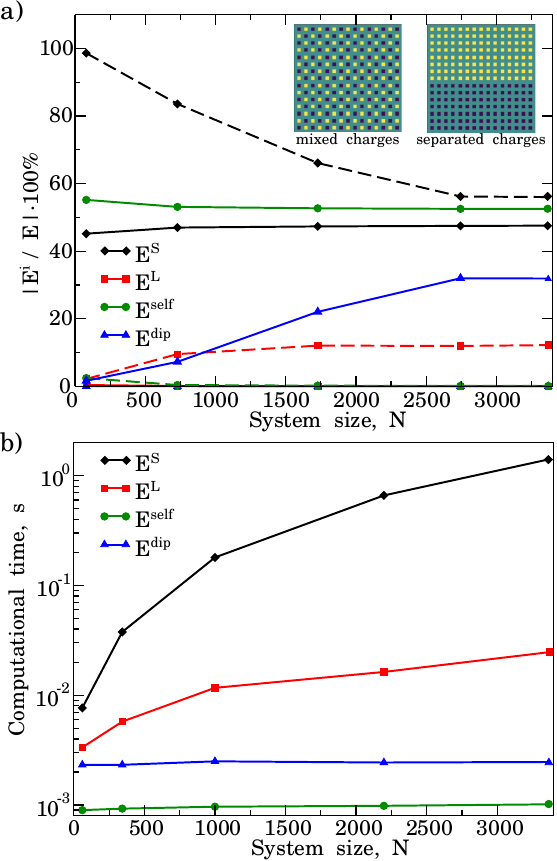} 
    \caption{Relative importance and computational time for electrostatic energy contributions (color denotation is explained in the graphs). $a)$ Contribution of different energy terms $E^{S}$, $E^{L}$, $E^{self}$, and $E^{dip}$ to the total electrostatic energy $E$ (Eq. (\ref{eq:eq5})) profiled as functions of the number of charges (system size) $N$. The dashed and solid lines correspond to charge configurations in which the positive charges (yellow pixels) and negative charges (black pixels) are mixed and separated, respectively (see {\it the insets} in panel $(a)$). $b)$ Computational time for calculation of $E^{S}$, $E^{L}$, $E^{self}$, and $E^{dip}$ as a function of $N$ mixed charges in a classical computer. Calculations were carried out in a 32$\times$32$\times$32 grid in the 3$d$-space. The cutoff distance was set to $R_{cut}=$ 9\AA.}
\end{figure}

\noindent where $\sigma$ is the standard deviation of the Gaussian function, ${\bf k}$ is the wavevector ($k=\|{\bf k}\|$), $V$ is the volume of the supercell, $S(\mathrm{\textbf{k}})$ is the structure factor, and $\epsilon'$ is the dielectric constant of the medium. The sum over ${\bf n}$ takes into account the periodic images of the charges and the prime indicates that in case $i = j$ the term ${\bf n}=0$ must be omitted. In Eq. (\ref{eq:eq5}), the summation in the first term $E^S$ is short-ranged in the real space (i.e. short-range electrostatics), since each term is truncated by the associated error function $\text{Erfc}(x)=1-\text{Erf}(x)$ ($\text{Erf}(x)$ is the error function defined by $\text{Erf}(x)=2/\sqrt{\pi}\int_0^x e^{-t^2}dt$). The summation in the second term $E^L$ is also short-ranged but in the reciprocal space (long-range electrostatics), since each term is truncated by the Gaussian function $e^{-\sigma^2k^2/2}$. To separate the total electrostatic energy ($E$) into the short-range contribution ($E^S$) and long-range contribution ($E^L$), typically, one uses the cutoff distance $R_{cut}=$ 9\AA. The third self-interaction energy term $E^{self}$ and the fourth dipole-interaction term $E^{dip}$ in Eq. (\ref{eq:eq5}) are not difficult to evaluate.

\subsection{Classical versus quantum algorithm}\label{sec:classical-ewald}

We propose a new algorithm for solving the Ewald summation problem on a quantum computer to evaluate the electrostatic energy $E$. We use the same decomposition of $E$ into several summations as in the Ewald method described above (Eq. (\ref{eq:eq5})). In this algorithm, the energy terms $E^{S}$, $E^{self}$, and $E^{dip}$ are computed using the classical algorithm, while the energy term $E^L$ is computed using the quantum computing to reduce the computational complexity. Fig. 2 shows that, depending on a charge distribution, $E^L$ accounts for $\sim$12$\%$ of the total energy $E$, and that calculation of $E^L$ is one of the computational bottlenecks. In the algorithm proposed, we utilize the Quantum Fourier Transform~\cite{nielsen2010quantum}. 

The computational complexity of the classical Ewald summation method \cite{ewald1921evaluation,de1980simulation} scales as $\bigO(N^{3/2})$ with the system size $N$. For a more advanced Particle Mesh Ewald (PME) method \cite{darden1993pme}, the computational complexity scales as $\bigO(N\log N)$. Therefore, we can formulate the following statement.

\begin{Theorem}
\label{thm:classic_Ewald}
In the Ewald summation method, $E^L$ is evaluated using fast Fourier transform, and the computational complexity is \[ T_C(N)= \bigO(N\log N).\]
\end{Theorem}

\begin{proof}
Let us choose parameter \(\sigma\) (see Eq. (\ref{eq:eq5})) in such a way that contributions from pairs of charges $q_i$ and $q_j$ located at \(\textbf{r}_{i}\) and \(\textbf{r}_{j}\) to \(E^{S}\) are negligible. This can always be done. For example, if the minimal distance \(r_{\text{min}}\) between charges is known, then \(\sigma\) can be chosen such that \(\sigma \ll r_{\text{min}}\). Then, the contributions to \(E^{S}\) from the terms with \(i \neq j\) are negligible because of an exponential decay in \(\mathrm{Erfc}\), and \(E^S\) only contains $N$ energy terms with \(i = j\). Then, the sum in \(E^{S}\) can be computed in \(\bigO(N)\) time. The sums in \(E^{self}\) and in \(E^{dip}\) can also be computed in \(\bigO(N)\) time. Next, to compute \(E^{L}\), the PME method~\cite{darden1993pme} uses interpolation of charges \(q_{i}\) on a grid. The total charge is described by a distribution
\[
Q(\mathbf{\ell}) = \sum_{j=1}^{N} q_{j} \theta_{p}(\mathbf{r}_{i,x},\mathbf{\ell}_{x}) \theta_{p}(\mathbf{r}_{i,y},\mathbf{\ell}_{y}) \theta_{p}(\mathbf{r}_{i,z},\mathbf{\ell}_{z}),
\]
where \(\theta_{p}\) are obtained from the weights of the \(p\)-th order Lagrangian interpolation  as
\[
\theta_{p}(x,k) = \phi_p(x-\kappa_p(x), k-\kappa_p(k)),
\]
with function \(\phi_p(x,k)\) defined by
\[
\phi_p(x,k) = \frac{(-1)^k \binom{2p-1}{k} \frac{1}{x-k/M}}{\sum_{l=0}^{2p-1}(-1)^l \binom{2p-1}{l}\frac{1}{x-l/M}}
\]
and with \(\kappa_p(x)\) being the integer function \(\kappa_p(x) = \lceil Mx \rceil - p + 1\).
Then, in the PME method the summation in \(E^{L}\) is expressed as a discrete convolution of \(Q(l)\) and an influence function
\begin{equation*}
\begin{split}
\Phi_{\mathrm{rec}}(\mathbf{f}; \sigma) &= \frac{1}{\pi V} \sum_{{\mathbf{m}\neq 0}\atop\mathbf{m}\in \mathbb{Z}^3} \frac{\exp(-\pi^{2}\sigma^{2}\mathbf{m}^{2})}{\mathbf{m}^{2}} \\
& \times \exp(2\pi i (\mathbf{m} \cdot \mathbf{f})).
\end{split}
\end{equation*}
which does not depend on the charge distribution $Q$. Since the convolution of $Q(l)$ with $\Phi_{\mathrm{rec}}$ can be computed using the fast Fourier transform in \(\bigO(N\log N)\) time, then \(E\) can be computed in \(\bigO(N) + \bigO(N\log N) + \bigO(N) + \bigO(N) = \bigO(N\log N)\) time.
\end{proof}

\subsection{Hybrid quantum-classical algorithm} \label{sec:hybrid-algo}

In the proposed algorithm for the calculation of electrostatic energy $E$, the QFT is used to compute the second term in Eq.~(\ref{eq:eq5}), $E^L$$=$$\frac{1}{2V\varepsilon_{0}}\sum_{\textbf{k$\neq$0}}\frac{e^{-\sigma^{2}k^{2}/2}}{k^{2}}|S(\mathrm{\textbf{k}})|^{2}$ (long-range electrostatics). The quantum Fourier transform algorithm provides an exponential advantage over the classical algorithm, by executing the Fourier transform on \(N\) points in \(\bigO((\log N)^{2})\) time. This time should be compared with \(\bigO(N \log N)\) time it takes the classical algorithms to carry out the same task (see previous section). However, the necessity of inputing data into a quantum computer limits somewhat the computational advantage obtained using the QFT algorithm. The information about the final quantum state can be obtained only through a measurement. Therefore, we must consider how many times the quantum transformation has to be performed (on copies of initial state) to reliably extract the result obtained, in order to preserve the quantum advantage.  

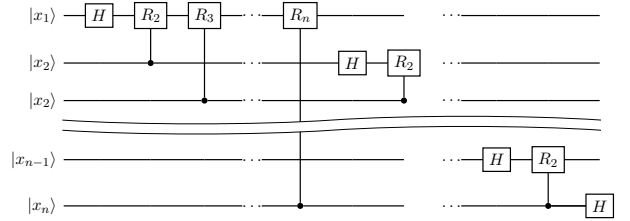
\begin{figure}[ht]
   \resizebox{\columnwidth}{!}{
       \begin{quantikz}
        \lstick{$|x_1\rangle$} & \gate{H} & \gate{R_2} & \gate{R_3} & \qw\ldots & \gate{R_n} & \qw      & \qw        & \ldots & \qw & \qw & \qw\\
        \lstick{$|x_2\rangle$} &  \qw     & \ctrl{-1}  & \qw        & \qw\ldots & \qw        & \gate{H} & \gate{R_2} & \ldots & \qw & \qw & \qw\\
        \lstick{$|x_2\rangle$} &  \qw     & \qw        & \ctrl{-2}  & \qw\ldots & \qw        & \qw      & \ctrl{-1}  & \ldots & \qw & \qw & \qw\\
        \wave &&&&& &&&&& &\\
        \lstick{$|x_{n-1}\rangle$}&\qw     & \qw       & \qw        & \qw\ldots & \qw        & \qw       & \qw        & \ldots & \gate{H} & \gate{R_2} & \qw\\
        \lstick{$|x_n\rangle$} &  \qw     & \qw        & \qw        & \qw\ldots & \ctrl{-5}  & \qw       & \qw        & \ldots & \qw     & \ctrl{-1}   & \gate{H}\qw
      \end{quantikz}
  }
  \caption{Quantum circuit for Quantum Fourier transform. The states $|x_1\rangle, |x_2\rangle, \ldots, |x_n\rangle$ comprise the computational basis set. The gate \(H\) denotes Hadamard gate, and \(R_n\) denotes the phase gate, which rotates the state by \(2\pi / 2^n\) radians around the \(z\)-axis (see Methods).}
  \label{fig:qft}
\end{figure}

\smallskip
\noindent
{\bf QFT algorithm complexity:} The QFT circuit (Fig. 3) that computes Fourier transform of an \(n\)-qubit register contains \(\bigO(n^2)\) gates: Hadamard gate \(H\) and \(n-1\) controlled rotations \(R_k\) for the first qubit, Hadamard gate and \(n-2\) controlled rotations for the second qubit, etc. Therefore, the total number of gates is \(\bigO(n^2)\). QFT operates with \(N = 2^n\) amplitudes, and the corresponding classical computation of fast Fourier transform requires \(\bigO(N \log N)\) operations, which is exponentially large. Similarly, the \(d\)-dimensional Fourier transform can be computed by applying \(d\) copies of the QFT circuit \cite{garcia-molina_rodriguez-mediavilla2022}. Therefore, the gate complexity of the \(d\)-dimensional version QFT is \(\bigO(dn^2)\).

\smallskip
\noindent
{\bf Quantum algorithm for computation of \(E^{L}\):} Let us focus on the \(E^{L}\)-part of electrostatic energy $E$ (see Eq.~(\eqref{eq:eq5}),
\[
E^{L} = \frac{1}{2V\varepsilon_{0}}\sum_{{\textbf{k$\neq$0}}\atop{\textbf{k}\in 2\pi\mathbf{Z}^3}}\frac{e^{-\sigma^{2}k^{2}/2}}{k^{2}}|S(\mathrm{\textbf{k}})|^{2},
\]
which is one of the computational bottlenecks (see Fig. 2b). In the expression for $E^{L}$, the structure factor 
\[
S(\mathbf{k}) = \sum_{j=0}^{N-1} q_{j} e^{i \mathbf{k} \cdot \mathbf{r}_{j}}
\]
(associated with a charge distribution) can be readily computed following {\it Steps 1--3} below: 

\smallskip
\noindent
{\it Step 1.} Generate discretization, normalize and encode a set of point charges \(q_{i}\) in a quantum register \[ \ket{\psi} = \frac{1}{\|\mathbf{q}\|}\sum_{j} q_{j}\ket{\mathbf{r}_{j}} ,\]
where \(\|q\| = \sqrt{\sum_{j=0}^{N-1} q_j^2}\) is the norm of the vector of charges.

\noindent
{\it Step 2.} Apply the multidimensional Fourier transform (Eq. {\eqref{eq:multi-qft}) to obtain state
\begin{align*}
\ket{\psi'}
&= \sum_{j} \frac{q_{j}}{M^{d/2}\|q\|} \Big( \sum_{\mathbf{k}} e^{2\pi i \mathbf{r}_{j} \cdot \mathbf{k}} \ket{\mathbf{k}} \Big) =\\
&= \sum_{\mathbf{k}} \Big( \sum_{j} \frac{q_{j}}{M^{d/2}\|q\|} e^{2\pi i \mathbf{r}_{j} \cdot \mathbf{k}} \Big) \ket{\mathbf{k}}
\end{align*}

\noindent
{\it Step 3.} Measure and register the resulting probabilities \(p_\mathbf{k} = \big| \sum_j \frac{q_j}{M^d\|q\|^2} e^{2\pi i \mathbf{r}_j \cdot \mathbf{k}} \big|^2\) for wave vectors \(\mathbf{k}\).

\subsection{Formalization} 

We place $N$ point charges on a cubic grid of size $M^3$. The positions of charges $q_j$, \(\mathbf{r}_{j}\), are given by \(\mathbf{r}_{jx} = L_{x}x_{j}/M\), \(\mathbf{r}_{jy} = L_{y}y_{j}/M\), \(\mathbf{r}_{jz} = L_{z}z_{j}/M\), where \(x_{j}\), \(y_{j}\) and \(z_{j}\) are the coordinates, and \(L_{x}\), \(L_{y}\) are \(L_{z}\) are the cell dimensions. Therefore, we can 
encode their locations \(\mathbf{r}_{j}\) in the 3D-space using \(3 \log M\) bits of information.

The quantum state $\ket{\psi}$ describes the point charges as follows. The basis states of the state
\[
\ket{\psi} = \sum_{j} \frac{q_{j}}{\|q\|}\ket{\mathbf{r}_{j}},
\]
are  \(3\log M\) qubits states of the form \(\ket{\mathbf{r}_{j}} = \ket{x_{j}}\ket{y_{j}}\ket{z_{j}}\), where  each coordinate is encoded by \(\log M\) qubits with $N$ nonzero amplitude coefficients $q_j / \|q\|$. Then, the multidimensional Fourier transform outputs the state
\begin{align*}
\ket{\psi'}
&= \sum_{j} \frac{q_{j}}{\|q\|} \Big( \frac{1}{M^{d/2}}\sum_{\mathbf{k}} e^{2\pi i \mathbf{r}_{j} \cdot \mathbf{k}} \ket{\mathbf{k}} \Big) =\\
&= \sum_{\mathbf{k}} \Big( \sum_{j} \frac{q_{j}}{M^{d/2}\|q\|} e^{2\pi i \mathbf{r}_{j} \cdot \mathbf{k}} \Big) \ket{\mathbf{k}}.
\end{align*}
Finally, we measure the state \(\ket{\psi'}\) to obtain the approximation of the structure factor \(S(\mathbf{k})\). We perform \(K\) repetitions of the same procedure, i.e. state initialization, multidimensional QFT, and measurement. Each repetition returns a measured value \(\ket{\mathbf{k}}\), so after \(K\) repetitions we have an approximation \(\tilde{p}_k\) of the resulting probability distribution: \(\tilde{p}_\mathbf{k} = c_\mathbf{k} / K \), where \(c_\mathbf{k}\) is the number of times we measured \(\ket{\mathbf{k}}\).

\subsection{Algorithm complexity}\label{sec:complexity-quantum}

\begin{Theorem}\label{thm:quantum_Ewald}
In the quantum implementation, $E^{L}$ is computed using QFT with the complexity 
\[
T_Q(K, M, N)=\bigO \big( Kd\log M( N + \log M) \big), 
\]
where $K$ is the number of repetitions, $d$ is the dimension of the system, $N$ is the system size, and $M$ is the grid size.
\end{Theorem}

\begin{proof}
To compute \(E^{L}\), we perform \(K\) repetitions of the quantum circuit encoding the charge distribution and performing multi-dimensional QFT. In the Gleinig-Hoefler method~\cite{Gleinig2021_sparse_states}, preparation of a sparse state of \(n\) qubits which has only \(S\) non-zero amplitudes requires \(\bigO(Sn)\) gates. In our case, the number of qubits \(n = d \log M\) and the state has \(S = N\) non-zero amplitudes, and so \(\bigO(Nd\log M)\) gates are required. The $1d$ QFT requires \((\log M)^2\) as a circuit depth (Fig. 3) because the number of qubits is \(n = \log M\). Therefore, we obtain the complexity \( T_{Q}(K, M, N) = \bigO \big( Kd( N\log M + (\log M)^2 )\big)\) quantum gates.
\end{proof}

The state representing the charge distribution has a few non-zero amplitudes. Therefore, it is more efficient to use special methods for this type of states. In our case, the number of qubits is \(n = \log M\) (and the number of charges is \(N\)), and so the generic (M\"ott\"onen) state preparation requires \(8M-4\log M-9\) gates, whereas the Gleinig-Hoefler method requires \(N\log M\) gates.
In our case, we have \(N < M^3\) (more precisely, \(N \ll M^3\)); hence, using the Gleinig-Hoefler method is more computationally efficient.

\begin{figure}[h]
    \centering
    \includegraphics[width=0.45\textwidth]{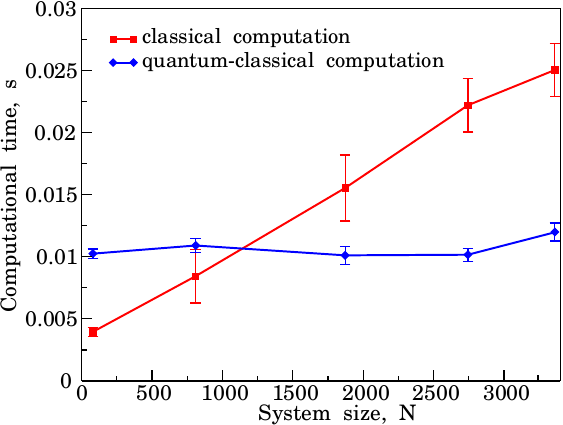} 
    \caption{The profiles of computational time (averages and standard deviations) for calculation of the $E^L$ term (see Eq. (\ref{eq:eq5})) using the classical algorithm and the quantum algorithm proposed (color denotation is explained in the graph). Calculations were carried out for the configuration of mixed charges (see the {\it inset} to Fig. 2) on a 32$\times$32$\times$32 grid in 3$d$-space. We used the cutoff distance $R_{cut}=$ 9\AA.}
\end{figure}

\noindent Therefore, in the Ewald summation approach to electrostatic energy calculation, the QFT provides saving in computational complexity. Comparing the statements in Theorems 1 and 2  and taking into account that we assume \(K\) (the number of repetitions) and \(d\) (the system dimension) to be constant, we see that we obtain the algorithm complexity \(\bigO(N\log M)\) for the quantum algorithm versus \(\bigO(N\log N)\) for the classical algorithm. We profiled the computational time for the classical and quantum-classical algorithms in Fig. 4, which shows that the quantum-classical algorithm has an advantage over the classical counterpart when $\log N$ exceeds $\log M$. To estimate the computational time on a quantum computer $T_{\mathrm{q}}$, we followed the procedure outlined in Ref~\cite{QML_Astrophysics}. We used a quantum computer emulator for which the computational time is $T_{\mathrm{em}}$. The quantum computational time is given by $T_{\mathrm{q}} = T_{\mathrm{em}} t_{\mathrm{q,1}} / t_{\mathrm{em,1}}$, where $t_{\mathrm{em,1}}$ is the time for single gate calculation on the emulator, and $t_{\mathrm{q,1}}$ = 50 ns is the time for single gate execution on a quantum computer (IQM reports $20$--$40$ ns)\cite{IQM}.

\section{Discussion and Conclusion}

Quantum computing holds an immense potential to tackle complex biophysical and biological
problems with exponentially large solution spaces~\cite{harris_kendon2010}, and can provide a computational advantage over traditional classical algorithms. This is in addition to quantum hashing approaches~\cite{Ablayev_Khadiev_Vasiliev_Ziiatdinov_2025_Quantum_Hashing}, which offer significant savings in quantum memory. Here, we proposed a quantum-classical algorithm for accurate calculation of electrostatic interactions for a system of point charges, in which the long-range interactions (energy $E^L$ in Eq. (\ref{eq:eq5})) are computed on a quantum device. The short-range interactions ($E^S$) are computed on a classical device along with the self-interaction and dipole interaction energies ($E^{self}$ and $E^{dip}$ in Eq. (\ref{eq:eq5})). 

A common approach to treat the long-range electrostatics is to employ the particle mesh Ewald method in conjunction with the periodic boundary conditions (PBC). Under PBC, a molecular system (e.g. biomolecule) is placed in a unit cell (solvation box), which is replicated in the $x-$, $y-$, and $z-$directions to generate image cells filling the 3d-space. The Ewald summation methods compute electrostatic interaction energy without truncating the energy at long distances, while also avoiding enumeration of all charge pairs $i$ and $j$. It is well known that infinite series in Eqs. (\ref{eq:eq3})-(\ref{eq:eq4}) are poorly converging; yet, employing the Ewald summation methods enables one to overcome this problem. 

The idea behind the Ewald approaches (i.e particle-particle-particle-mesh (P$^3$M) method~\cite{eastwood1980p3m3dp}, the particle mesh Ewald (PME) method~\cite{darden1993pme}, and the smooth particle mesh Ewald (SPME) method~\cite{essmann1995spme}) is to split the electrostatic energy into several energy components. The $E^S$-term in Eq. (\ref{eq:eq5}) represents interactions of point charges with other point charges partially screened by the screening charge clouds. The $E^L$-term describes interactions of point charges with the compensating charge clouds. Because the distribution of compensating charge clouds is periodic, $E^L$ can be calculated using Fourier transform in the reciprocal space. $E^{self}$ is associated with the interaction of point charges with their own compensating charge clouds, whereas $E^{dip}$ accounts for the dipolar contribution to the electrostatic energy. Since $E^S$, $E^{self}$ and $E^{dip}$ scale linearly with the number of charges ($N$), they can be computed on a classical device, while $E^L$ -- the Fourier part of the total electrostatic energy $E$, can be calculated on a quantum device. To compute $E^L$, which accounts for up to 12$\%$ of electrostatic energy (Fig. 2a) and presents a computational bottleneck (Fig. 2b), we utilized the power of Quantum Fourier transform.

\begin{figure}[h]
    \centering
    \includegraphics[width=0.45\textwidth]{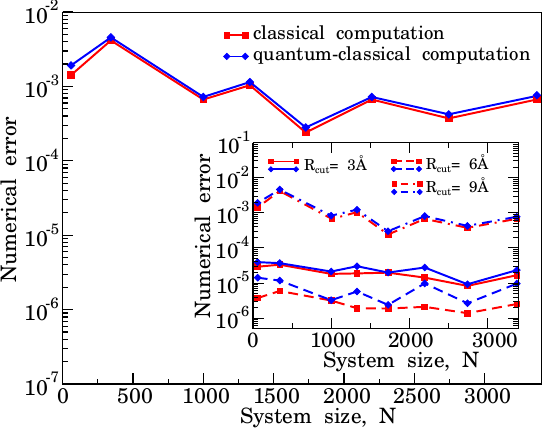} 
    \caption{Numerical accuracy for computation of electrostatic energy for a system of $N$ point charges using quantum-classical and classical algorithms. Profiled is a relative error ${(|E_{ex}-E_{app}|}) /{E_{ex}}$ (averages and standard deviations), where $E_{ex}$ is the exact electrostatic energy calculated using direct summation (Eq. (\ref{eq:eq3})), and $E_{app}$ is the approximate energy obtained with quantum-classical and classical algorithms (color denotation is explained in the graph). Calculations were performed for the configuration of mixed charges (see the {\it inset} to Fig. 2) on a 32$\times$32$\times$32 grid in the 3$d$-space (cutoff distance was set to $R_{cut}=$ 9\AA). {\it The inset} compares the profiles of relative error for variable cutoff distance $R_{cut}=$ 3, 6, and 9\AA.}
\end{figure}

We were able to demonstrate a quantum advantage of the quantum-classical algorithm proposed over the classical algorithm (Fig. 4). Because electrostatic interactions are the strongest interactions in condensed phase physical, chemical and biological systems, it is important to compute electrostatic energy (or force) contributions both efficiently and accurately. Therefore, we compared the numerical accuracy associated with the computation of electrostatic energy using the quantum-classical algorithm proposed and the classical algorithm by profiling the relative error ${(|E_{ex}-E_{app}|})/{E_{ex}}$ associated with the approximate treatment of electrostatic energy $E_{app}$ (using the quantum-classical and classical algorithms) and exact electrostatic energy $E_{ex}$ obtained through direct summation (Eq. (\ref{eq:eq3})). The results are displayed in Fig. 5, which shows that the numerical error below $10^{-3}-10^{-2}$. 

\begin{figure}[h]
    \centering
    \includegraphics[width=0.45\textwidth]{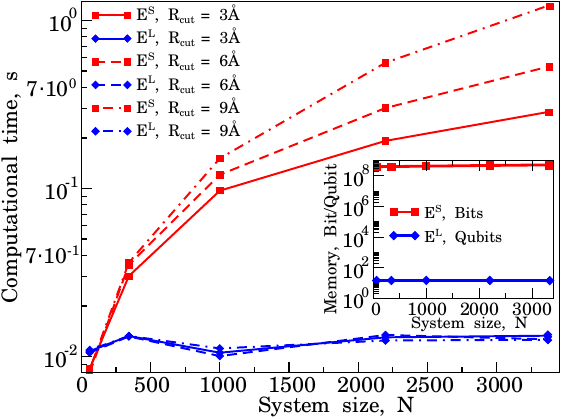} 
    \caption{Profiles of computational time and memory demand (the {\it inset}) associated with calculation of the short-range $E^S$ and long-range $E^L$ electrostatic energy terms for a system of $N$ point charges on a $32 \times 32 \times 32$ grid in the 3d-space (color denotation is explained in the graph) for variable cutoff distance $R_{cut}=$ 3, 6, and 9\AA. Calculations were performed for the configuration of mixed charges (see the {\it inset} to Fig. 2).}
\end{figure}

In the all-atom Molecular Dynamics schemes~\cite{jorgensen1996opls, brooks2009charmm, GROMOS, salomon2013amber}, the cutoff distance ($R_{cut}=$ 9-10\AA) is typically used to divide the total electrostatic energy ($E$) into the short-range contribution ($E^S$) and long-range contribution ($E^L$). The use of $R_{cut}$ is also related to generation of the Verlet lists of the nearest-neigbor atoms for different atoms. This choice of value for $R_{cut}$ (9-10\AA) is due to the fact that in condensed phases the inter-atomic (mainly electrostatic) interactions become weak at long separation distances ($>$9-10\AA). We profiled the dependence of numerical error on the cutoff distance for $R_{cut}=$ 3, 6, and 9\AA. The results obtained show that both for the quantum-classical and classical algorithms the numerical error is comparable and low ($10^{-5}-10^{-2}$), and that numerical error does not much vary with the system size $N$ (see the {\it inset} in Fig. 5) for $R_{cut}=$ 3-9\AA. 

In the context of quantum computing of electrostatic interactions, $R_{cut}$ can be viewed as a tunable parameter which can be used to vary the computational burden associated with the calculation of $E^S$- and $E^L$-energy terms. The shorter $R_{cut}$, the fewer terms are in the summation for $E^S$ and the more terms are in the summation for $E^L$ (see Eq. (\ref{eq:eq5})). Indeed, Fig. 6 shows that the computational time decreases for $E^S$ about 5-7-fold, yet barely changes for $E^L$ (because $E^L$ is computed on a quantum device) for $N=$2,000-3,500 point charges when $R_{cut}$ is decreased from 9 to 3\AA. At the same time, the amount of memory associated with the computation of $E^S$- and $E^L$-energy terms remains nearly constant (see the {\it inset} to Fig. 6). Hence, small values of $R_{cut}$ should be used to speedup the calculation of electrostatic interactions on a quantum computer. 

Quantum Fourier transform is an important tool used in a variety of scientific applications~\cite{wakeham_schuld_2024_qft_qml, shen2021qfcnn, vorobyov2021qft_sensing, farsian2025qft_astrophysics} in general, and in quantum algorithms~\cite{Shor_1997,garcia-molina_rodriguez-mediavilla2022}, in particular. The QFT-based algorithms are broadly used as quantum procedures in many quantum algorithms' applications~\cite{Ruiz_Perez_2017,Shakeel2020qft_qwalk,Arsoski2024qft_mcx}. The QFT algorithm provides an exponential advantage in computational complexity over the classical implementations of discrete quantum transform~\cite{nielsen2010quantum}. One of the well-known QFT algorithms is the Shor's quantum factorization algorithm \cite{Shor_1997}. In this study, we have successfully demonstrated that QFT can also be used to speedup computations of the electrostatic interactions in condensed phases, thereby expanding the scope of potential applications of QFT-based methods in theoretical biophysics, computational chemistry and biology.

To conclude, we developed and tested a new algorithm for both fast, yet, accurate computation of Coulomb electrostatic energy for a system of point charges on a quantum computer. The algorithm utilizes the Ewald-type splitting of electrostatic energy into various energy terms, of which ``the Fourier part'' of the electrostatic energy (long-range electrostatics) is computed using the Quantum Fourier transform. We have demonstrated the quantum advantage of the algorithm proposed (Fig. 4) over the classical algorithm for systems of charged particles when the number of 3$d$-grid points ($M^3$) exceeding the system size ($N$), and have assessed the algorithm's numerical accuracy by showing law overall numerical error (Fig. 5). The quantum advantage might be enhanced by calculating other electrostatic energy terms on a quantum computer (e,g, $E^S$ in Eq. (\ref{eq:eq5})). The accuracy of calculation of the Fourier component of electrostatic energy can be further improved by adopting the method of ``classical shadows''~\cite{Huang_2020}.

\section*{Data availability}
No datasets were generated or analyzed during the current study.

\section*{Acknowledgements}
M.Z. acknowledges financial support from the scholarship "Design and development of intelligent solutions for Industry 5.0" program number FTE0000382 -- CUP B47H22004430008 -- COR 22573728.

\section*{Author contributions}
M.Z. contributed to the development and implementation of the quantum algorithm on a quantum computer emulator, and wrote the paper. I.N. implemented the Ewald summation based algorithm, performed numerical calculations, and wrote the paper. F.A. formulated the problem, contributed to the development of quantum algorithm and wrote the paper. V.B. formulated the problem, contributed to the development of quantum algorithm and wrote the paper.

\section*{Conflicts of interest}
The authors declare no conflicts of interest.

\bibliographystyle{quantum}

\end{document}